\documentclass{ecai} 

\usepackage{latexsym}
\usepackage{amssymb}
\usepackage{amsmath}
\usepackage{amsthm}
\usepackage{booktabs}
\usepackage{enumitem}
\usepackage{graphicx}
\usepackage{color}

\usepackage{algorithm}
\usepackage{thm-restate}
\usepackage{mathtools}
\usepackage[noend]{algpseudocode}
\usepackage{amsfonts}
\usepackage{tikz}
\usetikzlibrary{positioning}
\usetikzlibrary{shapes.geometric}
\usepackage{subfigure}

\usepackage[english]{babel}
\usepackage[autostyle, english = american]{csquotes}
\MakeOuterQuote{"}

\tikzset{redvertex/.style={black, draw=black, circle, fill, scale=0.5}}
\tikzset{blkvertex/.style={black, draw=black, circle, scale=0.5}}

\newtheorem{theorem}{Theorem}
\newtheorem{lemma}[theorem]{Lemma}

\newtheorem{observation}[theorem]{Observation}

\newtheorem{problem}{Problem}

\newcommand{\BibTeX}{B\kern-.05em{\sc i\kern-.025em b}\kern-.08em\TeX}

\begin{document}


\begin{frontmatter}


\paperid{7449} 


\title{EFX Orientations of Multigraphs}


\author[A]{\fnms{Kevin}~\snm{Hsu}\orcid{0000-0002-8932-978X}\thanks{Corresponding Author. Email: kevinhsu996@gmail.com}}


\address[A]{University of Victoria}


\begin{abstract}
We study EFX orientations of multigraphs with self-loops. In this setting, vertices represent agents, edges represent goods, and a good provides positive utility to an agent only if it is incident to the agent. We focus on the bi-valued symmetric case in which each edge has equal utility to both incident agents, and edges have one of two possible utilities $\alpha > \beta \geq 0$. In contrast with the case of simple graphs for which bipartiteness implies the existence of an EFX orientation, we show that deciding whether a symmetric multigraph $G$ of any multiplicity $q \geq 2$ has an EFX orientation is NP-complete even if $G$ is bipartite, $\alpha > q\beta$, and $G$ contains a structure called a {\em non-trivial odd multitree} (NTOM). Moreover, we show that NTOMs are a problematic structure in the sense that even very simple NTOMs can fail to have EFX orientations, and multigraphs that do not contain NTOMs always have EFX orientations that can be found in polynomial-time.
\end{abstract}

\end{frontmatter}


\section{Introduction}\label{sec:intro}

The fair division problem is concerned with distributing resource or tasks among a group of agents in a fair manner. This problem has attracted widespread interest among researchers due to its applicability in wide-ranging situations, from splitting rent among housemates, course allocation among students, and chores division between household members.

We are particularly interested in the division of indivisible goods. Specifically, these are goods that cannot be divided or shared between different agents, such as single movie tickets or rooms in a shared living situation. Classical fairness notions such as envy-freeness (EF) and proportionality are often unattainable in this setting. Consider the classic example of allocating one good among two agents --- one agent is necessarily envious. Consequently, alternative fairness notions have been introduced. Among the most prominent of these is envy-freeness up to any good (EFX) \citep{caragiannis2019unreasonable,gourves2014near}. In an EFX allocation, we allow for each agent $i$ to envy another agent $j$, as long as the envy can be alleviated by disregarding any one good allocated to $j$. Since its inception, the existence problem of EFX allocations has evaded the attempts of many researchers and is one of the most important problems in this field, and has even been referred to as "fair division's most enigmatic question" by \citet{procaccia2020technical}.

Much of the work surrounding this problem focuses on special cases. For example, EFX allocations of goods are known to exist if the agents have identical utility functions or identical ranking of the goods~\citep{plaut2020almost}, if the agents have lexicographic utility functions~\citep{hosseini2023fairly}, or if the number of agents is at most 3~\citep{chaudhury2020efx,akrami2022efx}. Recently, \citet{christodoulou2023fair} initiated the study of graphical instances, which are instances representable by a graph in which vertices represent agents, edges represent goods, and a good can only have positive utility to an agent if its corresponding edge is incident to the agent's corresponding vertex. This captures the setting where agents only care about goods that are "close" to them, which can arise when the goods are located in physical space and agents have geographical constraints. In such situations, it is desirable for the goods to only be assigned to agents that have positive utility for them, so that none of them are "wasted". These types of allocations are called {\em orientations}, reflecting the fact that they exactly correspond to orientations of the underlying graph.

Although EFX orientations are desirable, their existence is harder to decide than EFX allocations. Indeed, \citet{christodoulou2023fair} showed that although EFX allocations always exist for simple graphs, deciding whether EFX orientations exist is NP-complete. \citet{deligkas2024ef1} further showed that this decision problem remains NP-hard even for graphs with a vertex cover of size 8. This leads one to ask whether there exist classes of graphs for which the existence of an EFX orientation can be decided efficiently.

Recent works have begun to address this. For example, \citet{zeng2024structure} found surprising connections between EFX orientations and the chromatic number of a graph, showing that EFX orientations exist for all bipartite simple graphs, and can only fail to exist for simple graphs $G$ whose chromatic number $\chi(G)$ is 3 or greater.

EFX allocations and orientations of multigraphs have also been studied. Most relevant to our work, \citet{kaviani2024almost} showed that symmetric multigraphs (i.~e.\ multigraphs in which each edge has equal utility to both endpoints) always have EFX allocations (Theorem 8.1 of \citep{kaviani2024almost}). Many other positive results have also been obtained --- EFX allocations of a multigraph $G$ exist if $G$ is bipartite \citep{afshinmehr2024efx,bhaskar2024efx,sgouritsa2025existence}, a multicycle \citep{afshinmehr2024efx}, a multitree \citep{bhaskar2024efx}, if $G$ has girth $2 \chi(G) -1$ \citep{bhaskar2024efx}, if each vertex has at most $\lceil |V(G)|/4 \rceil -1$ neighbours \citep{sgouritsa2025existence}, or if the shortest cycle with non-parallel edges has length at least 6 \citep{sgouritsa2025existence}.

On the other hand, deciding whether EFX orientations for multigraphs exist is NP-complete  because multigraphs are a generalization of simple graphs. In fact, \citet{deligkas2024ef1} showed that this problem remains NP-complete even for multigraphs $G$ containing only 10 vertices using a reduction from \textsc{Partition}. \citet{afshinmehr2024efx} studied the special case of bipartite multigraphs, and determined the relationship between the multiplicity $q$ and the diameter $d$ of bipartite multigraphs $G$ and their EFX orientability. Specifically, they showed that for bipartite multigraphs $G$, if (1) $G$ is acyclic, $q=2$, and $d \leq 4$, or if (2) $G$ is acyclic, $q>2$, and $d \leq 2$, then one can compute an EFX orientation of $G$ in polynomial time. In all other cases, $G$ may or may not have an EFX orientation.

\subsection{Our Contribution}

The existing work relating to fair orientations and allocations have been mostly limited to the setting of simple graphs. In this paper, we consider the more general setting of multigraphs. Our model allows for both parallel edges and self-loops. The fundamental premise that an agent can only have non-zero utility for its incident edges remains the same.

We find it useful to limit our focus to {\em bi-valued symmetric multigraphs}. Specifically, we make two assumptions. First, {\em symmetric} means that for each edge, both of its endpoints agree on the utility of the edge. In this case, we call the utility of an edge its {\em weight}. Second, {\em bi-valued} means that each edge is of one of two possible weights $\alpha > \beta \geq 0$ (i.~e.\ {\em heavy edges} and {\em light edges}). Although these conditions limit the generality of our results, similar assumptions have been made in existing literature (for example, symmetric valuations have been considered in \citep{christodoulou2023fair,afshinmehr2024efx,deligkas2024ef1,zeng2024structure} and bi-valued valuations have been considered in \citep{ABFHV21,feige2022maximin,garg2023computing,zeng2024structure}). In particular, the use of light and heavy edges appears in the proof of Lemma 3.1 in \citep{zeng2024structure}, although they do not refer to them explicitly as such. 

Our first contribution is a hardness result regarding EFX orientations. Since simple graphs are a special case of multigraphs of multiplicity $1$ (i.~e.\ the maximum number of parallel edges between a pair of vertices is $1$), the result due to \citet{zeng2024structure} can be restated as: bipartite multigraphs $G$ of multiplicity $1$ have EFX orientations. Thus, it seems plausible that bipartiteness is also a sufficient condition for multigraphs of higher multiplicities $q$ to have EFX orientations. As we will see, this is not the case for any $q \geq 2$, that is, there exist bipartite multigraphs that do not have EFX orientations. In contrast, \citet{afshinmehr2024efx} showed that bipartite multigraphs always have EFX allocations. Similarly, the problem of finding EFX orientations could be tractable if there are few values that utility functions can assume, or if those values fall within a small interval. Our result implies that the problem remains NP-hard under a restrictive setting that incorporates the above considerations. Moreover, it suggests a structure that we call a {\em non-trivial odd multitree} (NTOM) that guarantees the existence of EFX orientations when avoided.

We informally define some terminology. Given a bi-valued symmetric multigraph $G$ with two edge weights, a {\em heavy component} is a component of $G$ if one were to ignore the set of light edges. an NTOM is a multigraph that is a tree on at least two vertices if parallel edges are ignored, and every edge has odd multiplicity. We say that the set $S$ of heavy edges of a heavy component of $G$ {\em induces} an NTOM if the edge-induced submultigraph $G[S]$ is an NTOM.

\begin{restatable}{theorem}{restatemain}\label{thm:main}
	For any fixed $q \geq 2$, deciding whether a bi-valued symmetric multigraph $G$ of multiplicity $q$ has an EFX orientation is NP-complete, even if the following hold:
	\begin{enumerate}
		\item $G$ is bipartite;
		\item $\alpha > q\beta$;
		\item The heavy edges of each heavy component of $G$ induce an NTOM.
	\end{enumerate}
\end{restatable}

Theorem~\ref{thm:main} suggests that the existence of a heavy component whose heavy edges induce an NTOM is a barrier to finding EFX orientations. One might hope that EFX orientations can still be found as long as there are few NTOMs. Unfortunately, we are able to find a simple class of examples showing this not to be the case.

\begin{observation}\label{obs:one-tree-example}
	For any $q \geq 1$, there exists a multigraph of multiplicity $q$ containing a unique heavy component whose heavy edges induce an NTOM, that fails to have an EFX orientation.
\end{observation}
\begin{proof}
	Let $G$ be the multigraph shown in Figure~\ref{fig:one-tree-example}. By symmetry, $G$ has a unique orientation. By setting $\alpha > q\beta$, the vertex not receiving the heavy edge strongly envies the vertex that does.
\end{proof}

\begin{figure}
	\centering
	\begin{tikzpicture}
		\node[blkvertex] (1) at (0,0) {};
		\node[blkvertex] (2) at (1,0) {};
		\draw[thick] (1) to (2);
		\draw[dashed] (1) to [loop, in=135, out=225, looseness=15] node[left] {*} (1);
		\draw[dashed] (2) to [loop, in=45, out=315, looseness=15] node[right] {*} (2);
	\end{tikzpicture}
	
	\caption{Example of an NTOM with two vertices having no EFX orientations. The solid edge represents a heavy edge. The dashed self-loops marked with * each represents $q \geq 1$ light edges. }
	\label{fig:one-tree-example}
\end{figure}

On the other hand, we show that if $G$ does not contain such a heavy component, then it has an EFX orientation. Our constructive proof implies a polynomial-time algorithm for finding an EFX orientation for this case.

\begin{restatable}{theorem}{restateNoTrees}\label{thm:non-tree-thm}
	Let $G$ be a bi-valued symmetric multigraph of multiplicity $q \geq 1$. If $G$ does not contain a heavy component whose heavy edges induce an NTOM, then $G$ has an EFX orientation that can be found in polynomial time.
\end{restatable}

\subsection{Related Work}

We briefly survey related work surrounding fair allocations on graphs. For a more general survey of fair division, we refer the reader to \citet{amanatidis2023fair}.

Beyond the results pertaining to EFX orientations discussed in Section~\ref{sec:intro}, the popular EF1 criterion has also been considered. In contrast with EFX, where envy is alleviated by diregarding {\em any} good, EFX only requires envy to be alleviated by diregarding the {\em most valuable} good. \citet{deligkas2024ef1} showed that EF1 orientations of simple graphs always exist and designed a pseudopolynomial-time algorithm for finding such orientations. This contrasts with the fact that deciding whether EFX orientations exist is NP-complete \citep{christodoulou2023fair}.

The division of {\em chores} (i.~e.\ items with non-positive marginal utility) on graphs has also been considered. The goods and chores cases seem to be fundamentally different. Recall that \citet{christodoulou2023fair} showed that when there are only goods, deciding whether EFX orientations exist is NP-complete. In contrast, \citet{hsu2025polynomialmanuscript} found polynomial-time algorithms that decide whether EF1 and EFX orientations exist when there are only chores, and find such orientations whenever they exist. On the other hand, deciding whether EF1 and EFX orientations exist is NP-complete for multigraphs \citep{hsu2025polynomialmanuscript}. \citet{ijcai2024p338} explored the setting of mixed manna, in which each item can be a good or a chore independently of other items. The original definition of EFX allocations can be adapted to this setting by allowing for either a good to be removed from the envied agent's bundle or a chore to be removed from the envious agent's bundle. Additionally, one might decide to allow for the removed item to have zero marginal utility. These considerations result in four variants of EFX. \citet{ijcai2024p338} classified the complexity of deciding whether EFX orientations or allocations exist in a mixed manna setting, and finding one if it exists, for each of the four variants of EFX.

A related setting involving {\em restricted additive instances} has also been studied. In this setting, every agent has an additive utility function. Moreover, each item $g$ is associated with a number $v_g \geq 0$ such that for each agent $i$ and each item $g$, we have $u_i(g) \in \{0, v_g\}$. Thus, restricted additive instances generalize symmetric graphical instances, and can be interpreted as symmetric {\em hypergraphs} in a straightforward manner. The restricted additive setting models scenarios in which the goods being allocated have intrinsic value, but not every good can be allocated to every agent. For instance, \citet{ijcai2022p3} gave the example of fairly allocating food --- while the nutritional value of a food item is intrinsic to the food item itself, some people may have food allergies and hence derive zero utility from certain items. \citet{ijcai2022p3} showed that restricted additive instances always have EF2X allocations, which are a relaxation of EFX allocations in which no agent envies another if they ignore any {\em two} items allocated to agent they envy. \citet{kaviani2024almost} showed that restrictive additive multigraphs (which are equivalent to symmetric multigraphs) always have EFX allocations. 

\subsection{Organization}

We introduce the model in Section~\ref{section:prelim}. In Section~\ref{section:hardness}, we prove our hardness result (Theorem~\ref{thm:main}). In Section~\ref{section:EFX-orientations}, we prove Theorem~\ref{thm:non-tree-thm}. We conclude the paper with a brief discussion and an open problem in Section~\ref{section:discuss}.

\section{Preliminaries}\label{section:prelim}

\subsection{The Multigraph Model}

A multigraph $G$ consists of a vertex set $V$ and an edge set $E$. Each edge $e$ is associated with two vertices $u, v$ called its {\em endpoints}. We say $u, v$ are {\em adjacent}, and $e$ is {\em incident} to each of $u, v$. If $u=v$, then $e$ is called a {\em self-loop} at $u=v$. Two edges $e, e'$ are {\em parallel} if they share the same two endpoints. Every edge is parallel with itself. The {\em multiplicity} of an edge $e$ is the number of edges parallel with $e$. The {\em multiplicity} $q$ of $G$ is the maximum multiplicity of an edge.

An instance of the fair division problem is a pair $(G, u)$ where $G=(V,E)$ is a multigraph and $u$ is a collection of utility functions $u_i: \mathcal{P}(E) \rightarrow \mathbb{R}_{\geq 0}$, one corresponding to each vertex $i$. The vertices and edges of $G$ represent agents and goods, respectively. Throughout the paper, we reserve the symbols $n$ and $m$ to refer to the number of vertices/agents and edges/goods, respectively.

	For a subset $S$ of edges and a vertex $i$, we use $u_i(S)$ to represent the utility that $i$ derives from receiving the set $S$. If $S$ contains a single edge $e$, we write $u_i(e)$ in lieu of $u_i(\{e\})$ for brevity. We make two assumptions about $u_i$. First, each $u_i$ is {\em additive}, i.~e.\ for each subset $S \subseteq E$, we have $u_i(S) = \sum_{e \in S} u_i(e)$. Second, if an edge $e$ is not incident to a vertex $i$, then $u_i(e) = 0$.

An instance $(G, u)$ is said to be {\em symmetric} if $u_i(e) = u_j(e)$ whenever $i, j$ are both endpoints of an edge $e$. In this case, we refer to $u_i(e)$ as the {\em weight} of the edge $e$. By abusing terminology, we also say that the multigraph $G$ is symmetric in this case. If there are only two possible edge weights $\alpha > \beta \geq 0$, we say $G$ is a {\em bi-valued symmetric multigraph}. We reserve the symbols $\alpha, \beta$ for edge weights and call edges of weight $\alpha$ {\em heavy edges} and edges of weight $\beta$ {\em light edges}.

For a bi-valued symmetric multigraph $G$, we define a {\em heavy component} as a maximal vertex-induced submultigraph $K$ of $G$ such that between any two vertices $v \neq w$ of $K$, there exists a path between $v$ and $w$ that only contains heavy edges. A {\em multitree} is a multigraph that is a tree if each set of parallel edges is treated as a single edge. A multitree is {\em non-trivial} if it has more than one vertex and {\em odd} if every edge has odd multiplicity.

Given a multigraph $G$ and a set $S \subseteq E(G)$ of edges, we define the {\em edge-induced submultigraph} $G[S]$ of $G$ as the multigraph such that
\begin{itemize}
	\item $V(G[S]) \coloneqq \{v \in V(G) \mid v \text{ is an endpoint of an edge in } S\}$;
	\item $E(G[S]) \coloneqq S$.
\end{itemize}
For the sake of brevity, we say that the set $S$ of edges {\em induces a multitree} if $G[S]$ is a multitree. Similarly, we say that $S$ {\em induces a non-trivial odd multitree (NTOM)} if $G[S]$ is an NTOM.

\subsection{Fair Allocations}

Let $G = (V, E)$ be a multigraph. A {\em partial allocation} $\pi$ of $G$ is an $n$-tuple $(\pi_1, \pi_2, \dots, \pi_n)$ of pairwise disjoint subsets $\pi_i$ of $E$. If $\cup_{i \in [n]} \pi_i = E$, we call $\pi$ a {\em complete allocation} or simply an {\em allocation}. Fix a (possibly partial) allocation $\pi$ and two vertices $i \neq j$. We say $i$ {\em envies} $j$ if $u_i(\pi_i) < u_i(\pi_j)$. We say $i$ {\em strongly envies} $j$ if there is an edge $e \in \pi_j$ such that $u_i(\pi_i) < u_i(\pi_j \setminus \{e\})$. If no vertex envies another, we say $\pi$ is {\em envy-free} (EF). If no vertex strongly envies another, we say $\pi$ is {\em envy-free up to any edge (good)} (EFX).

We also define the notion of {\em private envy-freeness} (PEF). Let $i \neq j$ be two vertices and $E_{ij}$ be the set of edges between $i$ and $j$, not including self-loops. For a (partial) orientation $\pi$ of $G$ and two vertices $i$ and $j$, we say $\pi$ is PEF between $i$ and $j$ if $u_i(\pi_i \cap E_{ij}) \geq u_i(\pi_j \cap E_{ij})$ and $u_j(\pi_i \cap E_{ij}) \geq u_j(\pi_j \cap E_{ij})$, i.~e.\ $i$ and $j$ do not envy each other if we only consider the edges between them. Clearly, if $\pi$ is PEF between $i$ and $j$, then neither $i$ nor $j$ envies the other.

A (partial) allocation is called a {\em (partial) orientation} if there exists a (partial) orientation of $G$ such that $e \in \pi_i$ if and only if $e$ is directed toward $i$. Orientations are desirable because none of the edges are "wasted" by being given to a non-incident vertex to which the edge has zero utility. A (possibly partial) orientation $\pi$ is an {\em extension} of a partial orientation $\pi'$ if $\pi$ orients all of the edges oriented by $\pi'$ in the same way as $\pi'$, and possibly orients some additional edges.

An allocation $\pi$ is called a {\em Pareto improvement} over an allocation $\pi'$ if $u_i(\pi_i) > u_i(\pi'_i)$ for some vertex $i$ and $u_j(\pi_j) \geq u_j(\pi'_j)$ for each vertex $j$. If no allocation is a Pareto improvement of an allocation $\pi$, then we say $\pi$ is {\em Pareto optimal} (PO).

\subsection{Circuit Satisfiability}

A Boolean circuit is a directed acyclic graph with three types of vertices:
\begin{enumerate}
	\item {\em Input vertices} with zero in-degree that are labelled with either a unique Boolean variable or with \texttt{true};
	
	\item {\em Internal vertices} that are each labelled with AND, OR, or NOT;
	
	\item A designated internal vertex called the {\em output vertex}.
\end{enumerate}
Moreover, we require that each internal vertex labelled AND or OR to have exactly 2 in-neighbours, and each internal vertex labelled NOT to have a unique in-neighbour. Clearly, a Boolean circuit represents a Boolean formula in a natural way.

Given a Boolean circuit, \textsc{CircuitSAT} \citep{cook1971complexity} asks whether it is possible to assign Boolean truth values to the input vertices representing Boolean variables in such a way that the output vertex is true (i.~e.\ the Boolean formula represented by the Boolean circuit is satisfied).

\section{Proof of NP-Completeness}\label{section:hardness}

In this section, we prove Theorem~\ref{thm:main}. Similarly to \citet{christodoulou2023fair}, we also show a reduction from the NP-complete problem \textsc{CircuitSAT} \citep{cook1971complexity}.

We first prove a lemma concerning the multigraph $H$ shown in Figure~\ref{fig:subgadget}, which appears as an induced submultigraph in some of the gadgets we use in Theorem~\ref{thm:main}.

\begin{figure}
	\centering
	\begin{tikzpicture}
		\node[redvertex] (1) at (0,0) {};
		\node[redvertex] (2) at (1,0) {};
		\node[redvertex] (3) at (2,0) {};
		\node[redvertex] (4) at (3,0) {};
		\node[redvertex] (5) at (4,0) {};
		\node (l1) at (0, -0.25) {$v_1$};
		\node (l2) at (1, -0.25) {$v_2$};
		\node (l3) at (2, -0.25) {$v_3$};
		\node (l4) at (3, -0.25) {$v_4$};
		\node (l5) at (4, -0.25) {$v_5$};
		\draw[thick] (2) to[bend left=25] (3)
		(3) to[bend left=25] (4);
		\draw[dashed] (1) to node[above] {$e$} (2)
		(4) to node[above] {$e'$} (5)
		(2) to[bend right=25] (3)
		(3) to[bend right=25] (4);
	\end{tikzpicture}
	\caption{The multigraph $H$. Each solid (resp.\ dashed) edge represents a heavy (resp.\ light) edge.}
	\label{fig:subgadget}
\end{figure}

\begin{lemma}\label{lemma:subgadget}
	Let $G$ be a bi-valued symmetric multigraph that contains $H$ (see Figure~\ref{fig:subgadget}) as an induced submultigraph. If $\alpha > q\beta$ and none of $v_2, v_3, v_4$ are adjacent to a vertex in $V(G) \setminus V(H)$, then $e$ is directed toward $v_1$ or $e'$ is directed toward $v_5$ in any EFX orientation of $G$.
\end{lemma}
\begin{proof}
	Fix any EFX orientation of $G$ and suppose for contradiction that $e, e'$ are directed toward $v_2, v_4$, respectively. Consider the two heavy edges of $H$. Suppose neither are directed toward $v_3$. Because $v_3$ is not adjacent to a vertex in $V(G) \setminus V(H)$, the maximum utility it can receive is $2\beta$ if it receives all its incident light edges. Since $\alpha > q \beta \geq 2\beta$, the vertex $v_3$ envies $v_2$. Since both $e$ and the heavy edge between $v_2$ and $v_3$ are directed toward $v_2$, the vertex $v_3$ still envies $v_2$ even if we disregard $e$, that is, $v_3$ strongly envies $v_2$, a contradiction.
	
	Suppose instead the two heavy edges are both directed in the same direction with respect to Figure~\ref{fig:subgadget}. Without loss of generality, assume they are both directed rightward. In this case, $v_2$ envies $v_3$ regardless of whether it receives the light edge between $v_2$ and $v_3$. This is because $v_2$ is incident to exactly one heavy edge (directed toward $v_3$) and two light edges, as $v_2$ is not adjacent to any vertex in $V(G) \setminus V(H)$. Since $v_2$ does not strongly envy $v_3$ and $v_3$ receives the heavy edge between $v_2$ and $v_3$, all of the light edges incident to $v_3$ are directed away from $v_3$, causing $v_3$ to envy $v_4$. However, because $e'$ is directed toward $v_4$, the vertex $v_3$ strongly envies $v_4$, a contradiction.
	
	Otherwise, both heavy edges are directed toward $v_3$, and we have a contradiction because $v_2$ strongly envies $v_3$.
\end{proof}

\begin{figure}[tb]
	\subfigure[OR gadget]{
		\centering
		\begin{tikzpicture}
			\node[blkvertex] (w) at (0, 0) {};
			\node[redvertex] (v) at (-1, 0) {};
			\node[blkvertex] (a') at (-2, 0) {};
			\node[redvertex] (u) at (1, 0) {};
			\node[blkvertex] (b') at (2, 0) {};
			\node[redvertex] (a) at (-2, 1) {};
			\node[redvertex] (b) at (2, 1) {};
			\node[redvertex] (c) at (0, -1) {};
			\node[blkvertex] (v') at (-1, -1) {};
			\node[blkvertex] (u') at (1, -1) {};
			\node[blkvertex] (c') at (0, -2) {};
			\draw[thick]	(a) to node [right] {$x$} (a') 
			(b') to node [left] {$y$} (b)
			(v) to (w)
			(w) to (u)
			(v) to (v')
			(u) to (u')
			(c) to node [right] {$x \vee y$} (c');
			\draw[dashed]	(a') to (v)
			(w) to (c)
			(u) to (b')
			(a) to[bend right=70] (c')
			(b) to[bend left=70] (c');
		\end{tikzpicture}}
	
	\subfigure[NOT gadget of multiplicity $q \geq 2$]{
		\centering
		\begin{tikzpicture}
			\node[redvertex] (a) at (0, 0) {};
			\node[blkvertex] (u) at (1, 0) {};
			\node[redvertex] (v) at (2, 0) {};
			\node[blkvertex] (w) at (3, 0) {};
			\node[redvertex] (b) at (4, 0) {};
			
			\node[blkvertex] (a') at (0, -1) {};
			\node[redvertex] (u') at (1, -1) {};
			\node[blkvertex] (v') at (2, -1) {};
			\node[redvertex] (w') at (3, -1) {};
			\node[blkvertex] (b') at (4, -1) {};
			
			\node (la) at (0, 0.25) {$v_1$};
			\node (lu) at (1, 0.25) {$v_2$};
			\node (lv) at (2, 0.25) {$v_3$};
			\node (lw) at (3, 0.25) {$v_4$};
			\node (lb) at (4, 0.25) {$v_5$};
			
			\node (la') at (0, -1.25) {$w_5$};
			\node (lu') at (1, -1.25) {$w_4$};
			\node (lv') at (2, -1.25) {$w_3$};
			\node (lw') at (3, -1.25) {$w_2$};
			\node (lb') at (4, -1.25) {$w_1$};
			
			\draw[thick]	(a) to node[left] {$x$} (a')
			(b) to node[right] {$\neg x$} (b')
			(u) to[bend left=25] (v)
			(v) to[bend left=25] (w)
			(u') to[bend left=25] (v')
			(v') to[bend left=25] (w');
			\draw[dashed]	(a) to (u)
			(w) to (b)
			(a') to (u')
			(w') to (b')
			(u) to[bend right=25] (v)
			(v) to[bend right=25]  (w)
			(u') to[bend right=25] (v')
			(v') to[bend right=25] (w');
			
		\end{tikzpicture}}
	\subfigure[Duplication gadget]{
		\centering
		\begin{tikzpicture}
			\node[redvertex] (a) at (0, 1) {};
			\node[blkvertex] (b) at (0, 0) {};
			\node[redvertex] (c) at (1, 1) {};
			\node[blkvertex] (d) at (1, 0) {};
			\draw[thick]	(a) to node[left] {$x$} (b)
			(c) to node[right] {$x$}(d);
			\draw[dashed]	(a) to (d)
			(b) to (c);
		\end{tikzpicture}}
	
	\subfigure[TRUE gadget]{
		\centering
		\begin{tikzpicture}
			\node[blkvertex] (1) at (1, 0) {};
			\node[redvertex] (2) at (0.5, 0.866) {};
			\node[blkvertex] (3) at (-0.5, 0.866) {};
			\node[redvertex] (4) at (-1, 0) {};
			\node[blkvertex] (5) at (-0.5, -0.866) {};
			\node[redvertex] (6) at (0.5, -0.866) {};
			\node[redvertex] (7) at (2, 0) {};
			\node[blkvertex] (8) at (3, 0) {};
			\node[redvertex] (9) at (4, 0) {};
			\node (l1) at (1.25, -0.25) {$v_1$};
			\node (l2) at (0.75, 1.066) {$v_2$};
			\node (l3) at (-0.75, 1.066) {$v_3$};
			\node (l4) at (-1.3, 0) {$v_4$};
			\node (l5) at (-0.75, -1.066) {$v_5$};
			\node (l6) at (0.75, -1.066) {$v_6$};
			\node (l7) at (2, -0.25) {$v_7$};
			\node (l7) at (3, -0.25) {$v_8$};
			\node (l7) at (4, -0.25) {$v_9$};
			\draw[thick] (8) to node[above] {\texttt{true}} (9)
			(2) to[bend left=25] (3)
			(3) to[bend left=25] (4)
			(5) to (6)
			(1) to (7);
			\draw[dashed] (1) to (2)
			(4) to (5)
			(6) to node[left] {$*$} (1)
			(7) to (8)
			(2) to[bend right=25] (3)
			(3) to[bend right=25] (4);
		\end{tikzpicture}}
	\caption{The gadgets used in Theorem~\ref{thm:main}. Each solid (resp.\ dashed) edge represents a heavy (resp.\ light) edge, except for the dashed edge marked with * in (d) which represent $q$ light edges.}
	\label{fig:gadgets}
\end{figure}

We are now ready to prove Theorem~\ref{thm:main}.

\restatemain*
\begin{proof}
	Before giving our reduction, we first make two assumptions about instances of \textsc{CircuitSAT}. The first assumption is that an instance contains no AND gates, because AND gates can be simulated using a combination of NOT and OR gates. The second assumption is that an instance uses the TRUE gate at least once. Both of these assumptions can be made without loss of generality.
	
	Fix any $q \geq 2$ and two edge weights such that $\alpha > q\beta$. Given an instance $C$ of \textsc{CircuitSAT}, we represent each Boolean variable $x$ with a heavy edge $e_x$ between a black and a white vertex. We use the duplication gadget in Figure~\ref{fig:gadgets}(c) to duplicate the edge $e_x$ as many times as the literals $x$ and $\overline{x}$ appear in $C$. The other gadgets shown in Figure~\ref{fig:gadgets} show how to combine Boolean expressions together according to $C$ to construct a multigraph $G$.
	
	We show the properties given in the theorem hold for $G$. First, any time we combine the gadgets in Figure~\ref{fig:gadgets}, we do not introduce edges between vertices of the same colour, so $G$ is bipartite. On the other hand, we have $\alpha > q\beta$ by construction. Moreover, by inspection of the gadgets in Figure~\ref{fig:gadgets}, each heavy component of $G$ is an NTOM. Finally, $G$ has multiplicity $q$ because the TRUE gadget is used at least once.
	
	We claim that $C$ has a satisfying assignment if and only if $G$ has an EFX orientation. The truth value assigned to a variable corresponds to the orientation of the heavy edge representing that variable. Specifically, an assignment of \texttt{true} (resp.\ \texttt{false}) corresponds to the heavy edge being directed toward its black (resp.\ white) endpoint. For convenience, we will directly refer to a heavy edge as being {\em true} or {\em false} depending on its orientation. We proceed to show that each of the gadgets correctly corresponds to the stated Boolean functions. The OR and duplication gadgets are those used in the reduction due to \citet{christodoulou2023fair}, so we only need to consider the NOT and TRUE gadgets.
	
	Consider the NOT gadget. Fix any EFX orientation. We show that the edge $\neg x$ is false if and only if the edge $x$ is true. By symmetry, we only need to show one direction. Suppose the edge $x$ is true, i.~e.\ it is directed toward $v_1$. Then, $w_5$ envies $v_1$. Since $w_5$ does not strongly envy $v_1$, the light edge between $v_1$ and $v_2$ must be directed toward $v_2$. The submultigraph induced by the vertices $v_i$ is exactly $H$ in Figure~\ref{fig:subgadget}. Since none of $v_2, v_3, v_4$ are adjacent to a vertex in $V(G) \setminus V(H)$, Lemma~\ref{lemma:subgadget} implies the light edge between $v_4$ and $v_5$ is oriented toward $v_5$. If the edge $\neg x$ is oriented toward $v_5$, then $w_1$ would strongly envy $v_5$. Hence, the edge $\neg x$ is oriented toward $w_1$, i.~e.\ it is false.
	
	Moreover, we must show that regardless of the orientation of $x$, the NOT gadget has an EFX orientation. By symmetry, assume $x$ is directed toward $v_1$ without loss of generality. It is straightforward to check that the orientation in which every edge is oriented in a clockwise manner, except for the light edge between $v_2, v_3$ and the light edge between $w_2, w_3$ which are oriented counterclockwise, is EFX.
	
	Now consider the TRUE gadget. Fix any EFX orientation. We show that the edge between $v_8$ and $v_9$ is directed toward $v_9$, i.~e.\ it is true. Suppose instead that it is directed toward $v_8$. Because $v_9$ does not strongly envy $v_8$, the light edge between $v_7$ and $v_8$ is directed toward $v_7$. Because $v_1$ does not strongly envy $v_7$, the heavy edge between $v_7$ and $v_1$ is directed toward $v_1$. Since $v_7$ does not strongly envy $v_1$, all $q+1$ light edges incident to $v_1$ are directed away from $v_1$. The submultigraph induced by $v_1, v_2, \dots, v_5$ is $H$ and none of $v_2, v_3, v_4$ are adjacent to any vertex in $V(G) \setminus V(H)$. Hence, Lemma~\ref{lemma:subgadget} implies the light edge between $v_4$ and $v_5$ is directed toward $v_5$. Since the light edges between $v_1$ and $v_6$ are directed toward $v_6$, regardless how the heavy edge between $v_5$ and $v_6$ is directed, one of $v_5, v_6$ strongly envies the other, a contradiction. Thus, the edge between $v_8$ and $v_9$ is directed toward $v_9$.
	
	It remains to exhibit an EFX orientation of the TRUE gadget. Orient all of the edges on the path between $v_1$ and $v_9$ rightward in the direction of $v_9$.  The remaining edges all belong to a cycle of length 6. Orient all of them counterclockwise along the cycle, except for the light edge between $v_2$ and $v_3$, which is oriented in a clockwise manner. It is straightforward to verify that this is an EFX orientation.
\end{proof}

There are two notable differences between the multigraphs that our reduction and the reduction due to \citet{christodoulou2023fair} produce. First, while their reduction can result in a non-bipartite multigraph $G$ as a consequence of their NOT and TRUE gadgets, we circumvent this by designing new gadgets that ensure $G$ is bipartite. This is important because while \citet{zeng2024structure} showed bipartite {\em graphs} to have EFX orientations, our reduction implies finding EFX orientations for bipartite {\em multigraphs} is NP-hard. Second, by replacing the TRUE gadget due to \citet{christodoulou2023fair} with a novel TRUE gadget, we ensure that each vertex of $G$ is incident to at most only two light edges, rather than the three in their reduction. Doing so tightens the gap between $\alpha$ and $\beta$ required to guarantee NP-hardness from $\alpha >3\beta$ to $\alpha >2\beta$.

\section{EFX Orientations}\label{section:EFX-orientations}

In this section, we prove Theorem~\ref{thm:non-tree-thm}, which together with Observation~\ref{obs:one-tree-example} imply that an NTOM is a problematic structure preventing an EFX orientation from existing --- as long as a symmetric bi-valued multigraph does not contain an NTOM, then it has an EFX orientation. If it does contain an NTOM, then it is possible that no EFX orientations exist.

We assume multigraphs to be connected. Otherwise, we can apply Theorem~\ref{thm:non-tree-thm} to each component.

\subsection{Technical Lemmas}

Our approach to proving Theorem~\ref{thm:non-tree-thm} consists of three steps. First, we find a (possibly partial) EF orientation of a small subset of the edges of $G$ that satisfies certain properties (Lemmas~\ref{lemma:non-odd-multitree} and \ref{lemma:not-multitree}). Second, we apply a technical lemma (Lemma~\ref{lemma:extension-lemma}) to finish orienting almost all of the remaining edges while maintaining envy-freeness, but leaving a matching of light edges unoriented. Finally, we finish orienting the matching of light edges (Lemma~\ref{lemma:light_matching}).

As $G$ does not have a heavy component whose heavy edges induce an NTOM, each heavy component $K$ is one of 3 types:
\begin{enumerate}
	\item $K$ is non-trivial and the heavy edges of $K$ induce a non-odd multitree;
	\item $K$ is non-trivial and the heavy edges of $K$ do not induce a multitree;
	\item $K$ is trivial.
\end{enumerate}

First, we handle the simple case in which every heavy component of $G$ is of type 3 (i.~e.\ trivial). Observe that any vertex with a self-loop must be unenvied in any EFX orientation. So, increasing the weight of a self-loop in an EFX orientation preserves EFX. These observations together allow us to assume without loss of generality that all self-loops of $G$ have weight $\beta$. Thus, all edges (including self-loops) of $G$ have weight $\beta \geq 0$.

In the case that $\beta = 0$, any orientation is EFX. Suppose $\beta > 0$. Observe that the EFX condition is scale-invariant in the sense that scaling any agent's utility function by any positive multiplicative factor does not affect whether an allocation is EFX. Thus, we may assume $\beta = 1$ without loss of generality. So, for each vertex $i$ and each edge $e$, we have $u_i(e) \in \{0,1\}$. It is known that for such instances (called {\em binary instances}), a maximum Nash welfare (MNW) allocation is simultaneously EFX \citep{ABFHV21} and PO \citep{caragiannis2019unreasonable}, and can be found in polynomial time \citep{darmann2015maximizing,barman2018greedy}. PO allocations of a multigraph whose edges have positive weight are orientations, so this can be stated as:

\begin{theorem}\citep{darmann2015maximizing,barman2018greedy}\label{thm:only-trivial}
	Let $G$ be a bi-valued symmetric multigraph. If every heavy component of $G$ is trivial, then $G$ has an EFX orientation that can be found in polynomial time.
\end{theorem}

The remainder of this section handles the case in which $G$ contains a heavy component of type 1 or 2. The following two lemmas show the first step in orienting such components. A {\em (partial) orientation} of a heavy component $K$ of a multigraph $G$ is a (partial) orientation of $G$ that only orients edges of $K$.

\begin{lemma}\label{lemma:non-odd-multitree}
	If $K$ is a type 1 heavy component, then there exists a (possibly partial) EF orientation $\pi^K$ of $K$ and two vertices $v \neq w$ in $K$ such that
	\begin{enumerate}
		\item for each vertex $i \in K$, we have $u_i(\pi^K_i) \geq \alpha$;
		\item all edges between $v$ and $w$ are oriented except possibly one light edge;
		\item $\pi^K$ is PEF between $v$ and $w$; and
		\item at most one edge is oriented between each pair of vertices, except for the pair $v, w$.
	\end{enumerate}
	Moreover, $\pi^K$ can be found in polynomial time.
\end{lemma}
\begin{proof}
	Since $K$ is type 1, it is non-trivial and the heavy edges of $K$ induce a non-odd multitree. Let $v \neq w$ be a pair of vertices of $K$ with an even number $h$ of heavy edges in between and $\ell$ denote the number of light edges between $v, w$. We construct $\pi^K$ in two steps. First, orient $h/2$ heavy edges and $\lfloor \ell/2 \rfloor$ light edges between $v, w$ to each of $v, w$. Then, fix a spanning tree $T$ of $K$ consisting of only heavy edges. $T$ contains exactly one heavy edge $e$ between $v, w$, which is already oriented. Orient the tree edges of $T$ (except for $e$) in the direction away from $v, w$.
	
	Clearly, $\pi^K$ satisfies (1)--(4) and can be found in polynomial time. To see that $\pi^K$ is EF, observe that neither of $v, w$ envies the other by (3), and for any other pair of vertices $i \neq j$, (1) and (4) ensures neither envies the other.
\end{proof}

\begin{lemma}\label{lemma:not-multitree}
	If $K$ is a type 2 heavy component, then there exists a (possibly partial) EF orientation $\pi^K$ of $K$ such that
	\begin{enumerate}
		\item for each vertex $i \in K$, we have $u_i(\pi^K_i) \geq \alpha$; and
		\item at most one edge is oriented between each pair of vertices.
	\end{enumerate}
	Moreover, $\pi^K$ can be found in polynomial time.
\end{lemma}
\begin{proof}
	Since $K$ is type 2, it is non-trivial and the heavy edges of $K$ do not induce a multitree. Let $T$ be a spanning tree of $K$ that consists of only heavy edges. Since the heavy edges of $K$ do not induce a multitree, there is a heavy self-loop at some vertex $v$ or there are two vertices $v \neq w$ that are joined by a heavy edge and not joined by a tree edge of $T$. Orient all tree edges of $T$ in the direction away from $v$. Then, orient either a self-loop at $v$ toward $v$ in the first case, or exactly one heavy edge between $v, w$ toward $v$ in the second case.
	
	In either case, each vertex receives exactly one incident heavy edge, so (1) holds. The construction implies (2). Moreover, (1)--(2) imply $\pi^K$ is EF. On the other hand, it is clear that $\pi^K$ can be constructed in polynomial time.
\end{proof}

To proceed, we require the following technical result.

\begin{theorem}\citep{plaut2020almost}\label{thm:2vertices}
	For two vertices $i \neq j$ with additive utility functions, there exists an EFX allocation in which $i$ does not envy $j$. Moreover, such an allocation can be found in polynomial time.
\end{theorem}

Their idea is to compute an EFX allocation $\pi$ from the perspective of $j$, that is, an allocation of the items into two bundles such that $j$ would not envy $i$ regardless of which bundle $j$ receives. Letting $i$ pick a bundle of $\pi$ before $j$ yields an EFX allocation in which $i$ does not envy $j$.

We can now show Lemma~\ref{lemma:extension-lemma}, which allows us to extend the partial orientations resulting from Lemmas~\ref{lemma:non-odd-multitree} and \ref{lemma:not-multitree} while preserving envy-freeness. 

\begin{lemma}\label{lemma:extension-lemma}
	Let $G$ be a bi-valued symmetric multigraph of multiplicity $q \geq 1$ and $\pi' = (\pi'_1, \pi'_2, \dots, \pi'_n)$ be a (possibly partial) EF orientation of $G$. If $\pi'$ orients at most one edge between a pair of vertices $i \neq j$ and
	\begin{enumerate}
		\item $u_i(\pi'_i) \geq \alpha$ and $u_j(\pi'_j) \geq \alpha$; or
		\item $u_i(\pi'_i) \geq \beta$ and $u_j(\pi'_j) \geq \beta$ and there are only light edges between $i, j$.
	\end{enumerate}
	or if $\pi'$ orients no edges between $i \neq j$ and
	\begin{enumerate}
		\item[3.] $u_i(\pi'_i) \geq \alpha$; or
		\item[4.] $u_i(\pi'_i) \geq \beta$ and there are only light edges between $i, j$.
	\end{enumerate}
	then there exists an EF extension $\pi$ of $\pi'$ such that all edges between $i, j$ and self-loops at $i, j$ are oriented. Moreover, $\pi$ can be found in polynomial time.
\end{lemma}
\begin{proof}
	If $\pi'$ orients an edge between $i$ and $j$, then we assume without loss of generality that it orients it toward $i$. We construct the extension $\pi$ of $\pi'$ in two steps. First, we finish orienting all the edges between $i, j$ in a way that ensures neither $i$ nor $j$ envies the other. Second, we orient all self-loops at $i, j$ toward $i, j$, respectively. Note that this construction does not depend on whether $\pi'$ orients an edge between $i$ and $j$.
	
	Set $\pi = \pi'$ and let $U$ denote the set of edges between $i, j$ that are not oriented by $\pi'$. Theorem~\ref{thm:2vertices} implies the existence of an EFX allocation $X = (X_i, X_j)$ of $U$ to vertices $i, j$ in which $j$ does not envy $i$. Note that $X$ is an orientation because it allocates edges between $i, j$ to $i, j$. Let $\pi$ orient the edges in $U$ according to $X$, and $\pi_i, \pi_j$ denote the resulting bundles of vertices $i, j$, respectively.
	
	We claim that $i$ and $j$ do not envy each other. By the construction of $X$, we have $u_j(X_j) \geq u_j(X_i)$. If $\pi'$ does not orient an edge between $i$ and $j$, then $u_j(\pi_j) \geq u_j(X_j) \geq u_j(X_i) = u_j(\pi_i)$, so $j$ does not envy $i$. Otherwise, $\pi$' orients some edge $e$ between $i, j$ toward $i$ and one of (1)--(2) holds. If (1) holds, then $u_j(\pi'_j) \geq \alpha \geq u_j(e)$. If (2) holds, then $u_j(\pi'_j) \geq \beta \geq u_j(e)$. In either case, $u_j(\pi_j') \geq u_j(e)$, so
	\begin{align*}
		u_j(\pi_j) &= u_j(\pi'_j) + u_j(X_j) \\
		&\geq u_j(\pi'_j) + u_j(X_i) \\
		&\geq u_j(e) + u_j(X_i) = u_j(\pi_i)
	\end{align*}
	On the other hand, if $X_j = \emptyset$, then $\pi$ orients no edges between $i, j$ toward $j$, so $i$ does not envy $j$. Assume $X_j$ is nonempty and fix any edge $e \in X_j$. If (1) or (3) holds, then $u_i(\pi'_i) \geq \alpha \geq u_i(e)$. If (2) or (4) hold, then there are only light edges between $i, j$, so $u_i(\pi'_i) \geq \beta \geq u_i(e)$. In either case, we have $u_i(\pi'_i) \geq u_i(e)$, so
	\begin{align*}
		u_i(\pi_i) &= u_i(\pi'_i) + u_i(X_i) \\
		&\geq u_i(\pi'_i) + u_i(X_j \setminus \{e\}) \\
		&= u_i(\pi'_i) + u_i(X_j) - u_i(e) \\
		&\geq u_i(X_j) = u_i(\pi_j)
	\end{align*}
	Thus, $i$ and $j$ do not envy each other. It follows that $\pi$ can additionally orient any self-loops at $i, j$ toward $i, j$, respectively, without causing envy.
	
	Our construction relies on Theorem~\ref{thm:2vertices} to find an EFX allocation between two vertices in polynomial time, so $\pi$ can be found in polynomial time.
\end{proof}

Using the above, we can construct an EF orientation $\pi^G$ of $G$ in which all heavy edges are oriented and the set of unoriented edges is a matching in $G$. Recall that a {\em matching} in $G$ is a set of edges $M$ such that each vertex of $G$ appears as an endpoint of an edge of $M$ at most once.

\begin{lemma}\label{lemma:all_but_a_matching}
	Let $G$ be a bi-valued symmetric multigraph containing a non-trivial heavy component. If $G$ does not contain a heavy component whose heavy edges induce an NTOM, then there is a (possibly partial) EF orientation $\pi^G$ of $G$ such that
	\begin{enumerate}
		\item all heavy edges are oriented;
		\item the set of unoriented edges is a matching $M$ in $G$; and
		\item $\pi^G$ is PEF between all pairs of vertices with an unoriented edge in between.
	\end{enumerate}
	Moreover, $\pi^G$ can be found in polynomial time.
\end{lemma}
\begin{proof}
	At least one heavy component of $G$ is of type 1 or 2. For each type 1 heavy component $K$, Lemma~\ref{lemma:non-odd-multitree} implies there exists a (possibly partial) EF orientation $\pi^K$ of $K$ and two vertices $v \neq w$ in $K$ such that (1) for each vertex $i \in K$, we have $u_i(\pi^K_i) \geq \alpha$; (2) all edges between $v$ and $w$ are oriented except possibly one light edge; (3) $\pi^K$ is PEF between $v$ and $w$; and (4) at most one edge is oriented between each pair of vertices, other than the pair $v, w$.
	
	For each type 2 heavy component $K$, Lemma~\ref{lemma:not-multitree} implies there exists a (possibly partial) EF orientation $\pi^K$ of $K$ such that (1) for each vertex $i \in K$, we have $u_i(\pi^K_i) \geq \alpha$; and (2) at most one edge is oriented between each pair of vertices. 
	
	Let $\pi$ be the (possibly partial) orientation of $G$ that orients an edge $e$ if and only if some $\pi^K$ orients $e$, in which case $\pi$ orients $e$ in the same direction as $\pi^K$.
	
	
	We mark every vertex in a type 1 or type 2 heavy component as "processed" and every vertex in a trivial heavy component as "unprocessed". Observe that $u_i(\pi_i) \geq \alpha \geq \beta$ for each processed vertex $i$ and $u_j(\pi_j) = 0$ for each unprocessed vertex $j$. In the following, we "process" the unprocessed vertices while maintaining two invariant properties: (P1) $\pi$ is EF and (P2) every processed vertex receives at least $\beta$ utility.
	
	For an unprocessed vertex $j$ that has a processed neighbour $i$, let $\pi$ orient exactly one light edge between $i, j$ toward $j$, and mark $j$ as "processed". Doing so does not cause $i$ or $j$ to envy the other because (P2) holds for $i$, and maintains (P1)--(P2) because (P2) now holds for $j$. Since $G$ is connected, we can repeat this procedure until all vertices are "processed".
	
	We claim that for all pairs of vertices $i \neq j$ in $G$ except for pairs of special vertices $v \neq w$ in type 1 heavy components, one of the conditions in Lemma~\ref{lemma:extension-lemma} holds. If $i$ and $j$ belong to different heavy components, then $\pi$ orients at most one edge between them, and (P1)--(P2) ensure $u_i(\pi_i) \geq \beta$ and $u_j(\pi_j) \geq \beta$ and there are only light edges between $i, j$, which is condition (2) of Lemma~\ref{lemma:extension-lemma}. Otherwise, $i \neq j$ belong to the same heavy component (except for the special pairs $v \neq w$) so condition (1) of Lemma~\ref{lemma:extension-lemma} holds.
	
	We now apply Lemma~\ref{lemma:extension-lemma} to all pairs of vertices $i \neq j$ except the special pairs $v \neq w$ belonging to a type 1 heavy component. This yields an EF extension of $\pi^G$ such that all edges between $i, j$ and self-loops at $i, j$ are oriented. Since $\pi^G$ is EF, we can orient any remaining unoriented self-loops toward their respective endpoints without causing envy.
	
	Clearly, $\pi^G$ orients all heavy edges, and all light edges except for exactly one light edge between a special pair of vertices $v \neq w$ in each type 1 heavy component. Moreover, Lemma~\ref{lemma:non-odd-multitree}(3) guarantees $\pi^G$ to be PEF between every such pair $v \neq w$. Thus, (1)--(3) hold for $\pi^G$.
	
	We show $\pi^G$ can be found in polynomial time. Finding the heavy components of $\pi^G$ takes polynomial time using breadth-first search (BFS) \citep{moore1959shortest}. Orienting each heavy component of type 1 and 2 using Lemmas~\ref{lemma:non-odd-multitree} and \ref{lemma:not-multitree} takes polynomial time. Processing the vertices takes polynomial time using BFS. Finally, applying Lemma~\ref{lemma:extension-lemma} at most once to each pair of vertices takes polynomial time.
\end{proof}

We now orient the remaining light edges using Lemma~\ref{lemma:all_but_a_matching}.

\begin{lemma}\label{lemma:light_matching}
	Let $G$ be a bi-valued symmetric multigraph $G$ containing a non-trivial heavy component. If $\pi^G$ is a partial EF orientation satisfying (1)--(3) of Lemma \ref{lemma:all_but_a_matching}, then there exists an extension $\pi$ of $\pi^G$ that is an EFX orientation of $G$ that can be found in polynomial time.
\end{lemma}
\begin{proof}
	Let $M$ be the set of unoriented light edges of $\pi^G$ and $v_i, w_i$ denote the endpoints of each $e_i \in M$. By possibly renaming $v_i$ and $w_i$, we can assume for each $i$, if $\pi^G$ orients an edge not between $v_i, w_i$ toward one of $v_i, w_i$, it orients one such edge toward $v_i$. If $\pi^G$ orients such edges to both $v_i$ and $w_i$, then our assumption holds without renaming $v_i$ and $w_i$.
	
	We claim that the extension $\pi$ of $\pi^G$ that orients each edge $e_i$ toward $w_i$ is an EFX orientation of $G$. Since $\pi^G$ is EF, any envy experienced by a vertex in $\pi$ is caused by an edge that $\pi$ orients in addition to $\pi^G$. Hence, any envy in $\pi$ is between $v_i$ and $w_i$ for some $i$. Fix such a pair $v_i, w_i$. The orientation $\pi$ orients $e_i$ toward $w_i$, so $w_i$ does not envy $v_i$. It remains to show that $v_i$ does not strongly envy $w_i$.
	
	Suppose first $\pi^G$ does not orient an edge not between $v_i, w_i$ toward one of $v_i, w_i$. Since $v_i$ does not envy $w_i$ in $\pi^G$ and $\pi$ only orients a single additional light edge between $v_i$ and $w_i$ to $w_i$, the envy that $v_i$ experiences toward $w_i$ can be alleviated by ignoring any of the light edges between $v_i$ and $w_i$. Because all edges oriented to $w_i$ by $\pi$ are between $v_i, w_i$, the vertex $v_i$ does not strongly envy $w_i$.
	
	Otherwise, $\pi^G$ orients an edge not between $v_i, w_i$ toward one of $v_i, w_i$. By our assumption, $\pi^G$ orients one such edge $e$ toward $v_i$. Since $\pi^G$ is PEF between $v_i$ and $w_i$, the vertices $v_i$ and $w_i$ receive an equal number of heavy (resp.\ light) edges that are between $v_i$ and $w_i$ in $\pi^G$. On the other hand, because $\pi^G$ orients $e$ toward $v_i$, we have $u_{v_i}(\pi^G_{v_i}) \geq u_{v_i}(\pi^G_{w_i}) + u_{v_i}(e)$. Thus, after $\pi$ orients the additional light edge $e_i$ between $v_i$ and $w_i$ toward $w_i$, we have
	\begin{align*}
		u_{v_i}(\pi_{v_i}) &= u_{v_i}(\pi^G_{v_i}) \\
		&\geq u_{v_i}(\pi^G_{w_i}) + u_{v_i}(e) \\
		&= u_{v_i}(\pi_{w_i}) - u_{v_i}(e_i) + u_{v_i}(e) \\
		&\geq u_{v_i}(\pi_{w_i})
	\end{align*}
	where the last inequality follows is because $u_{v_i}(e) \geq u_{v_i}(e_i)$ as $e_i$ is a light edge. So, $v_i$ does not envy $w_i$.
	
	Clearly, $\pi$ can be found in polynomial time given $\pi^G$.
\end{proof}

\subsection{Proof of Theorem~\ref{thm:non-tree-thm}}

\restateNoTrees*
\begin{proof}
	If every heavy component of $G$ is trivial, then $G$ has an EFX orientation that can be found in polynomial time by Theorem~\ref{thm:only-trivial}. Otherwise, $G$ contains a non-trivial heavy component. By Lemma~\ref{lemma:all_but_a_matching}, there exists a (possibly partial) EF orientation $\pi^G$ of $G$ such that all heavy edges are oriented, the set of unoriented edges is a matching $M$ in $G$, and $\pi^G$ is PEF between all pairs of vertices with an unoriented edge in between. By Lemma~\ref{lemma:light_matching}, there exists an extension $\pi$ of $\pi^G$ that is an EFX orientation of $G$ that can be found in polynomial time.
\end{proof}

\section{Discussion}\label{section:discuss}

In this paper, we showed that deciding whether a bi-valued symmetric multigraph $G$ of multiplicity $q$ has an EFX orientation is NP-complete, even in a highly restrictive setting. We also showed that as long as $G$ does not contain a heavy component whose heavy edges induce an NTOM, then $G$ has an EFX orientation that can be found in polynomial time. It is important to note that our results do not preclude such multigraphs from having EFX allocations. Indeed, \citet{kaviani2024almost} showed that all symmetric multigraphs have EFX allocations.

On the other hand, one might ask whether there are additional settings in which deciding whether EFX orientations exist is possible in polynomial time. Theorem~\ref{thm:main} suggests that it may be fruitful to consider settings in which $\alpha$ and $\beta$ are not so far apart.

\begin{problem}
	Is it possible to decide in polynomial time whether bi-valued symmetric multigraphs for which $\alpha \leq q\beta$ have EFX orientations?
\end{problem}




\begin{ack}
	We acknowledge the support of the Natural Sciences and Engineering Research Council of Canada (NSERC), funding reference number RGPIN-2022-04518. We also thank the anonymous reviewers for their constructive comments. 
\end{ack}



\bibliography{bibliography}

\end{document}